\documentclass[10pt]{IEEEtran}
\usepackage{calc}

\usepackage[normalem]{ulem}
\usepackage{url}
\usepackage{hyperref}
\hypersetup{colorlinks,urlcolor=blue}

\hypersetup{colorlinks=false,pdfborder=000}

\makeatletter
\DeclareUrlCommand\ULurl@@{%
  \def\UrlLeft{\uline\bgroup}%
  \def\UrlRight{\egroup}}
\def\ULurl@#1{\hyper@linkurl{\ULurl@@{#1}}{#1}}
\DeclareRobustCommand*\ULurl{\hyper@normalise\ULurl@}

\usepackage{xcolor,multirow,gensymb}
  \usepackage{cite}
\usepackage{graphicx,subfigure,epsfig}
\usepackage{psfrag}
\usepackage{amsmath,amssymb}
\usepackage{color}
\usepackage{array}
\usepackage{nicefrac}
\usepackage{amsmath,amssymb}
\usepackage{xifthen}
\usepackage{mathtools}
\usepackage{enumerate}
\usepackage{microtype}
\usepackage{xspace}
\usepackage{bm}
\usepackage[T1]{fontenc}
\usepackage{fancyhdr}
\usepackage{lastpage}
\usepackage{rotating}
\usepackage{tabulary}

\renewcommand{\deg}{\degree}

\newcommand{\mbs}[1]{\bm{#1}}
\newcommand{\mat}[1]{{\uppercase{\mbs{#1}}}}
\newcommand{\Id}{\mat{\mathrm{I}}}

\newcommand{\T}{{\scriptscriptstyle\mathsf{T}}}
\renewcommand{\H}{{\scriptscriptstyle\mathsf{H}}}

\renewcommand{\Re}[1][]{\ifthenelse{\isempty{#1}}{\operatorname{Re}}{\operatorname{Re}\left(#1\right)}}
\renewcommand{\Im}[1][]{\ifthenelse{\isempty{#1}}{\operatorname{Im}}{\operatorname{Im}\left(#1\right)}}

\def\var{{\mathrm{var}}}

\def\bA{{\mathbf{A}}}

\def\bD{{\mathbf{D}}}
\def\bE{{\mathbf{E}}}
\def\bF{{\mathbf{F}}}
\def\bG{{\mathbf{G}}}
\def\bH{{\mathbf{H}}}
\def\bI{{\mathbf{I}}}

\def\bR{{\mathbf{R}}}

\def\bY{{\mathbf{Y}}}
\def\bZ{{\mathbf{Z}}}

\newcommand{\cC}{{\cal C}}

\newcommand{\cN}{{\cal N}}

\def\rp{{\mathrm{p}}}

\def\bee{{\mathbf{e}}}
\def\bff{{\mathbf{f}}}
\def\bg{{\mathbf{g}}}
\def\bh{{\mathbf{h}}}

\def\bs{{\mathbf{s}}}

\def\bw{{\mathbf{w}}}
\def\bx{{\mathbf{x}}}
\def\by{{\mathbf{y}}}
\def\bz{{\mathbf{z}}}
\def\b0{{\mathbf{0}}}

\def\bTheta{{\boldsymbol{\Theta}}}

\def\bbC{{\mathbb{C}}}

\def\bpsi{{\boldsymbol{\psi}}}
\def\bPsi{{\boldsymbol{\Psi}}}

\newcommand{\EE}{\mathbb{E}}

\newcommand{\CN}[1][]{\ifthenelse{\isempty{#1}}{\mathcal{N}_{\mathbb{C}}}{\mathcal{N}_{\mathbb{C}}\left(#1\right)}}

\renewcommand{\P}[1][]{\ifthenelse{\isempty{#1}}{\mathbb{P}}{\mathbb{P}\left(#1\right)}}
\newcommand{\E}[1][]{\ifthenelse{\isempty{#1}}{\mathbb{E}}{\mathbb{E}\left(#1\right)}}
\renewcommand{\det}[1][]{\ifthenelse{\isempty{#1}}{\text{det}}{\text{det}\left(#1\right)}}
\newcommand{\trace}[1][]{\ifthenelse{\isempty{#1}}{\text{tr}}{\text{tr}\left(#1\right)}}
\newcommand{\rank}[1][]{\ifthenelse{\isempty{#1}}{\text{rank}}{\text{rank}\left(#1\right)}}
\newcommand{\diag}[1][]{\ifthenelse{\isempty{#1}}{\text{diag}}{\text{diag}\left(#1\right)}}
\def\nn{\nonumber}

\newtheorem{proposition}{Proposition}
\newtheorem{remark}{Remark}
\newtheorem{corollary}{Corollary}

\newcommand{\hatvg}{\hat{\bg}} 
 
\def\bPhi{{\boldsymbol{\Phi}}}
  
\newcommand{\tr}{\mathop{\mathrm{tr}}\nolimits}
\newtheorem{Theorem}{Theorem}

\newcommand{\al}{\alpha}

\newcounter{enumi_saved}
\usepackage{answers}
\Newassociation{solution}{Solution}{solutionfile}

\AtBeginDocument{\Opensolutionfile{solutionfile}[\jobname]}
\AtEndDocument{\Closesolutionfile{solutionfile}\clearpage
}

\begin{document}

\title{Downlink Performance of Massive MIMO under General Channel Aging Conditions}
\author{Anastasios K. Papazafeiropoulos \\
Communications and Signal Processing Group, Imperial College London, London, U.K.\\
Email: a.papazafeiropoulos@imperial.ac.uk
\thanks{This research was supported by a Marie Curie Intra European Fellowship within the $7$th European Community Framework Programme for Research of the European Commission under grant agreement no. [$330806$], IAWICOM.}}
\maketitle
\begin{abstract}
Massive multiple-input multiple-output (MIMO) is a promising technology aiming at achieving high spectral efficiency by deploying a large number of base station (BS) antennas  using coherent combining. Channel aging due to user mobility is a significant degrading factor of such systems. In addition, cost efficiency of massive MIMO is a prerecuisite for their deployment, that leads to low cost antenna elements inducing high phase noise. Since phase is  time-dependent, it contributes to channel aging. For this reason, we present a novel joint channel-phase noise model, that enables us to study the downlink of massive MIMO with maximum ratio transmission (MRT) precoder under these conditions by means of the deterministic equivalent of the achievable sum-rate. Among the noteworthy outcomes is that the degradation due to user mobility dominates over the effect of phase noise. Nevertheless, we demonstrate that  the  joint effects of phase noise and user mobility do not degrade the power scaling law $1/\sqrt{M}$ ($M$ is the number of BS antennas), as has been established in massive MIMO systems with imperfect channel state information.
\end{abstract}
\section{Introduction}
The large multiple-input multiple-output (MIMO) paradigm, known also as massive MIMO~\cite{Rusek}, is a promising technology for achieving higher throughput per area, which is necessary for next generation systems~\cite{Larsson}. In particular, this concept employs multi-user MIMO technniques, where the base station (BS) includes a large antenna array that results to the vanishment of both inter-user interference and thermal noise  as the number of antennas increases due to channel orthogonality.

Originally, most studies assume perfect  channel state information (CSI); however, this is a strong assumption, being untenable in practice. Given that the performance depends on the quality of CSI at the transmitter, it is of paramount importance to investigate realistic scenarios with imperfect CSI, acquired during the uplink training phase in time-division-duplex (TDD) architectures~\cite{Rusek,Hoydis}.

One of the main sources contributing to imperfect knowledge of  CSI is channel aging, present in time-varying channels, that describes the mismatch between the estimated and current channels due to the relative movement between users and BS antennas. Despite its crucial role, it has not been studied in depth. Specifically, the uplink and downlink signal-to-interference-and-noise-ratios (SINRs) have been derived  by using a deterministic equivalent analysis after applying a maximum ratio combining (MRC) decoder and  a maximum ratio transmission (MRT) precoder, respectively~\cite{Truong}. Furter investigation and comparison took place in~\cite{PapazafeiropoulosWCNC,PapazafeiropoulosICASSP}, where more sophisticated linear techniques were applied. Moreover, the optimal linear receiver for cellular systems has been derived in ~\cite{PapazafeiropoulosICC} by exploiting the correlation between the channel estimates and the interference from other cells, and in~\cite{PapazafeiropoulosPIMRC}, the uplink analysis of a cellular network with zero-forcing (ZF) receivers for any finite as well as infinite number of BS antennas has been provided. Neverteless, tight lower bounds and the power scaling law have been studied for massive MIMO with channel aging in~\cite{PapazafeiropoulosISIT}.

In practice, the application of massive MIMO is meaningful, if both hardware cost and power consumption are kept low. In such case, hardware impairments become a significant degrading factor leading to imperfect CSI. Among the hardware imperfections, phase noise~\cite{Demir}, randomly phase shifting the output of the local oscillator (LO), accumulates over time. Unfortunately, it has been scarcely studied in the case of massive MIMO systems, and mostly for uplink scenarios. Specifically, the phase noise has been taken into account for single carrier  uplink  massive MIMO by considering time-reversal MRC and ZF in~\cite{Pitarokoilis1} and~\cite{Pitarokoilis2}, respectively. Furthermore, in~\cite{Bjornson1}, the achievable uplink user rate with MRC has been obtained by considering a general model with additive and multiplicative (phase noise) distortions, and it was shown that the phase noise can be alleviated in the large antenna limit. Especially, regarding the downlink,  an effort to address the effect of phase noise, being of great importance in large MIMO systems,  has taken place only in~\cite{KrishnanMassive}. 

Given that phase noise, being time dependent, contributes to channel aging, this work provides a novel model incorporating  both effects of user mobility and phase noise. In fact, we present an insightful joint channel-phase noise model. In addition, we provide the deterministic equivalent of the downlink  SINR with MRT. Furthermore, the power scaling law has been studied in the case of user mobility combined with phase (general channel aging conditions).  Numerical results reveal the impact of both effects in massive MIMO for various mobility conditions and topologies of LOs. Among the interesting observations, we mention that the effect of user mobility is more severe than phase noise, and the massive MIMO technique keeps being favorable even in general time-varying conditions.
\section{System Model}
We consider a single-cell system with a BS serving  $K$  single-antenna non-cooperative users. In particular, the BS is deployed with an array of $M$ antennas. Our study focuses on large scale networks, where both the number of antennas $M$ and users $K$ grow infinitely large, while keeping a finite ratio $\beta$, i.e., $M, K\to \infty$ with $K/M=\beta$ fixed such that $\mathrm{lim\,sup}_{M} K/M < \infty$.

A realistic quasi-static fading model is taken into account, where the $M\times 1$ channel response vector between the BS and user $k$ during the $n$th symbol, is denoted by $\bh_{k,n}  \in \bbC^{M}$. Especially, $\bh_{k,n}  \in \bbC^{M}$ exhibits flat fading, and it is assumed constant during the current symbol,  but it may vary slowly from symbol to symbol.  Moreover,  common effects such as path loss dependent on the distance of the users from the BS, and the spatial correlation due to limited antenna spacing, are considered by expressing the channel between  as
\begin{align}
 \bh_{k,n} = \bR^{1/2}_{k}\bw_{k,n},
\end{align}
where $\bR_{k} =\mathbb{E}\left[ \bh_{k,n} \bh_{k,n}^{\H}\right]\in \bbC^{M \times M}$ is a deterministic Hermitian-symmetric positive definite matrix representing the aforementioned effects.  The independence of $\bR_{k}$ by $n$ is a reasonable assumption, since effects such as spatial correlation and lognormal shadowing vary with time in a much slower pace than the coherence time~\cite{Truong,PapazafeiropoulosWCNC,PapazafeiropoulosICC}. In addition, $\bw_{k,n} \in \bbC^{M}$ is an uncorrelated fast fading Gaussian channel vector drawn as  a realization from zero-mean circularly symmetric complex Gaussian distribution, i.e. $\bw_{k,n} \sim \cC\cN(\b0,\bI_{M})$. 

\subsection{Phase Noise Model}
In practice,  phase noise, induced during the up-conversion of the baseband signal to bandpass and vise-versa, impairs both  the transmitter and the receiver, and it appears as a multiplicative factor to the channel vector. The cause is a distortion in phase due to the random phase drift of the generated carrier by the corresponding LO. 

Our analysis regarding the $n$th time slot,  takes into account for both synchronous and non-synchronous operation at the receiver. The latter setting assumes independent phase noise processes $\phi_{m,n}, m=1,\ldots,M$ with $\phi_{m,n}$ being the phase noise process at the $m$th antenna. Note that the  phase noise processes are considered as mutually independent, since each antenna has its own oscillator, i.e., separate independent LOs (SLOs) at each antenna.  This setting can degenerate to identical phase noise processes, if the LOs are assumed identical (ILOs). In the synchronous scenario, all the BS antennas are connected with a common LO (CLO), i.e., there is only one phase noise process $\phi_{n}$. Furthermore, we denote $\varphi_{k,n}$, $k=1,\ldots, K$ the phase noise process at the single-antenna user $k$. 

Phase noise during the $n$th symbol can be described by a discrete-time  independent Wiener process , i.e., the phase noises at the  LOs of the $m$th antenna of the BS and $k$th user are modeled as 
\begin{align}
 \phi_{m,n}=\phi_{m,n-1}+\delta^{\phi_{m}}_{n}\label{phaseNoiseBS}~~\mathrm{and}~~
 \varphi_{k,n}=\varphi_{k,n-1}+\delta^{\varphi_{k}}_{n},
\end{align}
where $\delta^{\phi_{m}}_{n}\sim \cN(0,\sigma_{\phi_{m}}^{2}) $ and $\delta^{\varphi_{k}}_{n}\sim \cN(0,\sigma_{\varphi_{k}}^{2})$. Note that $\sigma_{i}^{2}=4\pi^{2}f_{\mathrm{c}} c_{i}T_{\mathrm{s}}$, $i=\phi_{m}, \varphi_{k}$ describes the phase noise increment variance with $T_{\mathrm{s}}$, $c_{{i}}$, and $f_{\mathrm{c}}$ being the  symbol interval, a constant dependent on the oscillator, and the carrier frequency, respectively. 

\subsection{Channel Estimation}
In real systems obeying to the TDD architecture, the BS  estimates the uplink channel~\cite{Ngo1,Hoydis,PapazafeiropoulosWCNC}.  The estimation  occurs by means of  an uplink training phase during each transmission block,  where the transmission of pilot symbols takes place. The knowledge of the downlink channel is considered known due to the property of reciprocity.

By assuming that the channel estimation takes place at time $0$, we now derive the joint channel-phase noise linear minimum mean-square error detector  (LMMSE) estimate of the effective channel $\bg_{k,0}=\bTheta_{k,0}\bh_{k,0}$ under phase noise, small-scale fading, and channel impairments such as path loss and spatial correlation\footnote{Since the duration of the training phase is small, the phase noise can be considered constant because its innovation is  negligible.}. Note that $\bTheta_{k,0}=\mathrm{diag}\left\{ e^{j \theta_{k,0}^{(1)}}, \ldots, e^{j \theta_{k,0}^{(M)}} \right\}$ is the phase noise because of the BS and user $k$ LOs at time $0$ with $ \theta_{k,0}^{(i)}=\phi_{i,0}+\varphi_{k,0},~i=1,\ldots, M$. 
\begin{proposition}
 The LMMSE estimator of $\bg_{k,0}$, obtained during the training phase, is
 \begin{align}\label{estimatedChannel}
 \hat{\bg}_{k,0}=\left(  \Id_{M}+\frac{\sigma_{b}^{2}}{p_{\rp}}\bR_{k}^{-1}\right)^{-1}\tilde{\by}_{,k,0}^{\rp},
\end{align}
where $\tilde{\by}_{,k,0}^{\rp}$ is a noisy observation of the effective channel from user $k$ to the BS, and $p_{\mathrm{p}}=\tau p_{\mathrm{u}}$  with $p_{\mathrm{u}}$ and $\tau$ being the power per user in the uplink data transmission phase and the duration of the  training sequence during the training phase, respectively.
\end{proposition}
\begin{proof}
During the training phase $T$, all users  transmit  simultaneous  mutually orthogonal pilot sequences consisting of length $\tau$ symbols. Also, we assume that the effective channel remains constant during this phase.  In particular,  the pilot sequences can be represented by  $\bPsi = [\bpsi_{1}; \cdots; \bpsi_{K}] \in \bbC^{K \times \tau}$ with $\bPsi$ normalized, i.e., $\bPsi\bPsi^\H = \bI_K$. Given that the channel estimation takes place at time $0$, the received signal at the BS is written as
\begin{align}\label{eq:Ypt}
\bY_{\rp,0}= & \sqrt{p_{\rp}}\bTheta_{k,0} \bH_{0}\bPsi + \bZ_{\rp,0},
\end{align}
where $p_{\rp}$ is the common average transmit power for all users, and $\bZ_{\rp,0} \in \bbC^{M \times \tau}$ is spatially white additive Gaussian noise matrix at BS   during this phase. By correlating the received signal with the training sequence $\frac{1}{\sqrt{p_{\rp}}}\bpsi^\H_k$ of user $k$, the BS obtains
\begin{align}
\tilde{\by}_{k,0}^{\rp}
=  \bTheta_{k,0} {\bh_{k,0}}  + \frac{1}{\sqrt{p_{\rp}}} \tilde{\bz}_{\rp,0},\label{eq:Ypt3}
\end{align}
where $\tilde{\bz}_{\rp,0}= \bZ_{\rp,0}\bpsi^\H_{k}\sim \cC\cN(\b0,\sigma_{b}^{2}\bI_{M})$. Since we seek to find the joint channel-phase noise estimate, we denote the effective channel at time $0$ to be $\bg_{k,0}=\bTheta_{k,0} {\bh_{k,0}}$. After applying the MMSE estimation method~\cite{Kay}, the  effective channel estimated at the BS  can be written as in~\eqref{estimatedChannel}.

Employing the orthogonality principle, the current channel decomposes as
\begin{align}
\bg_{k,0} = \hat{\bg}_{k,0} + \tilde{\bg}_{k,0},\label{eq:MMSEorthogonality}
\end{align}
where $\hat{\bg}_{k,0}$ is distributed as $\cC\cN\left(\b0,\bD_{k}\right)$ with $\bD_{k}=\left(  \Id_{M}+\frac{\sigma_{b}^{2}}{p_{\rp}}\bR_{k}^{-1}\right)^{-1}\bR_{k}$, and $\tilde{\bg}_{k,0} \sim \cC\cN(\b0, \bR_{k} - \bD_{k})$ is the channel estimation error vector. Note that $\hat{\bg}_{k}$ and $\tilde{\bg}_{k}$ are statistically independent because they are uncorrelated and jointly Gaussian.
\end{proof}

\subsection{Channel Aging}
The relative movement between the users and the BS results to a Doppler shift, which makes the channel vary between when it is learnt by estimation, and when the estimate is applied for detection/precoding. The higher the velocity, the greater the performance loss because of worse channel estimation. The description of channel aging can be made by  an autoregressive model of order $1$, which relates the  joint channel-phase noise process $\bg_{k,n}$ between the BS and the $k$th user at time $n$ (current time) to the estimated channel during the training phase. Specifically, we  have
\begin{align}
\bg_{k,n}  = \bA_{n} \bg_{k,0} + \bee_{k,n},\label{eq:GaussMarkoModel}
\end{align}
where, initially, $\bA_{n}=\mathrm{J}_{0}(2 \pi f_{D}T_{s}n) \Id_{M}$ is supposed to model 2-D isotropic scattering~\cite{Jakes}\footnote{Normally, the BS antennas are deployed in a fixed space instead of considering  distantly distributed small BSs constituting one large MIMO BS. In such case $\bA_{n}$ is a scalar, since it is reasonably assumed that all antennas appear the same relative movement comparing to the user.}, $\bg_{k,0}$ is the channel obtained during the training phase, and $\bee_{k,n} \in \bbC^{M}$ is the uncorrelated channel error vector due to the channel variation modelled as a stationary Gaussian random process with i.i.d.~entries and distribution $\cC\cN(\b0,\bR_{k}-\bA_{n}\bR_{k}\bA_{n})$. 

Since~\eqref{eq:MMSEorthogonality} allows to expresss the  channel at time $0$  by means of its estimate,  the effective channel at time $n$ can be  written as
\begin{align}
\bg_{k,n} 
= & \bA_{n} \bg_{k,0} + \bee_{k,n}\nn\\
= & \bA_{n}\hat{\bg}_{k,0} + \tilde{\bee}_{k,n},\label{eq:GaussMarkov2}\end{align} 
where $\tilde{\bee}_{k,n}=\bA_{n} \tilde{\bg}_{k,n} +  \bee_{k,n}\sim \cC\cN(\b0, \bR_{k} - \bA_{n}\bD_{k}\bA_{n})$ and $\hat{\bg}_{k,0}$ are mutually independent. As far as $\bA_{n}$ is concerned, it is assumed known by the BS. Interestingly, it incorporates  the effects of 2-D isotropic scattering, user mobility,  and phase noise, as shown by Theorem~\ref{LMMSE}. Consequently, useful outcomes can be extracted  during the  analysis in Section~\ref{Deterministic}.  
 
\begin{Theorem}\label{LMMSE}
 The estimated effective joint channel-phase noise channel vector at time $n$  is expressed as
 \begin{align}
  {\bg}_{k,0}^{\mathrm{Joint}}=\bA_{n} {\bg}_{k,0}
 \end{align}
with 
\begin{align}
\bA_{n}=\mathrm{J}_{0}(2 \pi f_{D}T_{s}n)e^{-\frac{\sigma_{\varphi_{k}}^{2}}{2}n}\Delta\bPhi_{n}, \label{agingParameter}
\end{align}
where $\Delta\bPhi_{n}=\mathrm{diag}\left\{e^{-\frac{\sigma_{\phi_{1}}^{2}}{2}n} ,\ldots,e^{-\frac{\sigma_{\phi_{M}}^{2}}{2}n}\right\}$.
 \end{Theorem}
\begin{proof}
The total mean square error (MSE), which is actually  given by means of the error covariance matrix ${\bE}_{k,n}$, is $\mathrm{MSE}=\tr {\bE}_{k,n}$. Specifically, ${\bE}_{k,n}$ can be written as
\begin{align}
{\bE}_{k,n}&=\mathbb{E}\left[ \left( \bg_{k,n}-\bA_{n} {\bg}_{k,0} \right) \left( \bg_{k,n}-\bA_{n} {\bg}_{k,0} \right)^{\H}\right]\nn\\
&=\mathbb{E}| \bg_{k,n}|^{2}+\bA_{n} \mathbb{E}|{\bg}_{k,0}|^{2}\bA_{n}^{\H}-2\mathbb{E}\left[ \mathrm{Re}\left\{\bA_{n}  {\bg}_{k,0}\bg_{k,n}^{\H} \right\} \right]\nn\\
&=\bR_{k}+\bA_{n}\bR_{k}\bA_{n}^{\H}-2\mathbb{E}\left[ \mathrm{Re}\left\{\bA_{n}  {\bg}_{k,0}\bg_{k,n}^{\H} \right\} \right]\nn\\
&=\bR_{k}+\bA_{n}\bR_{k}\bA_{n}^{\H}-2 {\mathrm{J}_{0}}(2 \pi f_{D}T_{s}n)e^{-\frac{\sigma_{\varphi_{k}}^{2}}{2}n}\bA_{n} \bR_{k}\Delta\bPhi_{n},\label{error_Cov_mat}
\end{align}
where  we have used that $\bg_{k,n}=\bTheta_{k,n}\bh_{k,n}$ with $\bTheta_{k,n}$ and $\bh_{k,n}$ being uncorrelated. In addition, we have denoted $\Delta\bPhi_{n}=\mathrm{diag}\left\{e^{-\frac{\sigma_{\phi_{1}}^{2}}{2}n} ,\ldots,e^{-\frac{\sigma_{\phi_{M}}^{2}}{2}n}\right\}$, and $\EE[{\bh}_{k,0}{\bh}_{k,n}^{\H}] = \mathrm{J}_{0}(2 \pi f_{D}T_{s}n)\bR_{k}$.
 $\bA_{n}$ is determined by minimizing the MSE in~\eqref{error_Cov_mat}. Thus, if we differentiate~\eqref{error_Cov_mat} with respect to $\bA_{n}$, and equate the resulting expression to zero,  we obtain the Hermitian  matrix $\bA_{n}$ as in~\eqref{agingParameter}.
\end{proof}

\begin{remark}
Notably, phase noise induces an extra loss in the coherence due to its accumulation over time between the actual channel and the estimated channel coefficients. In other words, phase noise is another contribution to the channel aging phenomenon that imposes a further challenge to investigate the realistic potentials of large MIMO.
\end{remark}

\begin{corollary}
 In case that  all BS antennas are connected with the same oscillator (CLO) or  all the oscillators are identical (ILO), $\bA_{n}$ degenerates to a scalar $\al_{n}= \mathrm{J}_{0}(2 \pi f_{D}T_{s}n)e^{-\frac{\sigma_{\varphi_{k}}^{2}+\sigma_{\phi_{k}}^{2}}{2}n}$ or a scaled identity matrix $\al_{n}\Id_{M}$, respectively. 
\end{corollary}

\subsection{Downlink Transmission}
In dowlink, the BS broadcasts data to its serving users with the same power $p_{\mathrm{d}}$. Exploiting that in TDD  the downlink channel is given by  the Hermitian  transpose of the uplink channel, the received signal $y_{k,n}\in\bbC$ by user $k$ during the data transmission phase ($n=\tau+1,\ldots, T_{\mathrm{c}}$) is
\begin{align}
\!\!y_{k,n}&=\sqrt{p_{\mathrm{d}}}\bh^\H_{k,n}\bTheta_{k,n}\bs_{n}+z_{k,n}\label{eq:DLreceivedSignal}
\end{align}
where  $\bs_{n}=\sqrt{\lambda}\bF_{n}\bx_{n}$ denotes the transmit signal vector by the  BS with $\lambda$, $\bF_{n} \in \bbC^{M \times K}$, and  $\bx_{n} = \big[x_{1,n},~x_{2,n},\cdots,~x_{K,n}\big]^\T \in \bbC^{K}\sim \cC\cN(\b0,\bI_{K})$ being a normalization parameter, the linear precoding matrix,  and the data symbol vector to its $K$ serving users, respectively. Also, $z_{k,n} \sim \cC\cN(0,\sigma_{k}^{2})$ is complex Gaussian noise at user $k$.  The normalization parameter is obtained by  constraining the transmit power per user to $p_{\mathrm{d}}$, i.e., $\mathbb{E}\left[\frac{p_{\mathrm{d}}}{K}\bs_{n}^{\H}\bs_{n}\right]=p_{\mathrm{d}}$, as
\begin{align}
\lambda=\frac{1}{\EE \left[\frac{1}{K} \tr\bF_{n}\bF^\H_{n}\right]}. \label{eq:lamda} 
\end{align}

Since $\bg_{k,n}=\bTheta_{k,n}\bh_{k,n}$,  the received signal is written as 
\begin{align}
&\!\!y_{k,n}\!=\!{\sqrt{\lambda p_{\mathrm{d}}}}\bg^\H_{k,n}\bTheta_{k,n}^{2} {\bff}_{k,n}x_{k,n}\nn\\
&\!\!+\sum_{i \neq k} {\sqrt{\lambda p_{\mathrm{d}}}}\bg^\H_{k,n}\bTheta_{k,n}^{2} { \bff}_{i,n}x_{i,n}+z_{k,n}\label{eq:DLreceivedSignal1}\\
&\!\!=\!\underbrace{{ \sqrt{\lambda p_{\mathrm{d}}}}\mathbb{E}\left[\bg^\H_{k,n} \bTheta_{k,n}^{2} {\bff}_{k,n}\right]\!x_{k,n}}_{\mathrm{desired~signal}}\!+\!\underbrace{z_{k,n}}_{\mathrm{noise}}\nn\\
&\!\!+\!\sum_{i \neq k} \!{\sqrt{\lambda p_{\mathrm{d}}}}\bg^\H_{k,n} \bTheta_{k,n}^{2} {\bff}_{i,n}x_{i,n}\nn\\
&\!\!+\!{ \sqrt{\lambda p_{\mathrm{d}}}}\left(\bg^\H_{k,n}\bTheta_{k,n}^{2} {\bff}_{k,n}x_{k,n}-\mathbb{E}\left[\bg^\H_{k,n} \bTheta_{k,n}^{2} {\bff}_{k,n}\right]x_{k,n}\right).\label{eq:DLreceivedSignalgeneral}
\end{align}

We reasonably use the above technique because user $k$ does not have access on instantaneous CSI, but he can be aware of $\EE[\bg^\H_{k,n} \bff_{jm} ]$. Moreover, since the available CSI at time $n$ is $\bA_{n}\hat{\bg}_{k,0}$,   as can be seen from~\eqref{eq:GaussMarkov2}, the MRT precoder is $ \bF_{n}=  \bA_{n}{\hat{\bG}}_{0}$.

By taking into account that the input symbols are Gaussian, and by considering the worst case uncorrelated additive noise with zero mean as in~\cite{Hassibi}, we consider a SISO model, scaled by $1/M^{2}$, with desired signal power
\begin{align}
S_{k,n}
&= \frac{\lambda}{M^{2}}  \Big|\EE\left[{\bg}^\H_{k,n}\widetilde{\bTheta}_{k,n}   \bA_{n} \hat{\bg}_{k,0} \right]\!\!\Big|^{2}.\label{normalizedMRT}
\end{align}
and interference plus noise power at user $k$ given by 
\begin{align}
I_{k,n}
&=   \frac{\lambda}{M^{2}}\var\left[\bg^\H_{k,n}\widetilde{\bTheta}_{k,n} \bff_{k,n}\right]  +   \frac{\sigma_{k}^{2}}{p_{\mathrm{d}}M^{2}} \nn\\
&+    \sum_{i\ne k}    \frac{\lambda}{M^{2}}\EE\bigg[\Big|\bg^\H_{k,n}\widetilde{\bTheta}_{k,n} \bff_{i,n}\Big|^{2}\bigg], \label{eq:DLgenIntfPower}
\end{align}
where we have denoted  $\widetilde{\bTheta}_{k,n}=\Delta{\bTheta}_{k,n}$. Note that we have assumed that $\bg_{k,n}=\bTheta_{k,0}^{}\bg_{k,n}$, since $\bg_{k,n}$ is circularly symmetric.

The mutual information between the detected symbol and the transmitted symbol is lower bounded by the achievable rate per user given by
\begin{align}
R_{k}=\frac{1}{T_{c}}\sum_{n=1}^{T_{c}-\tau}\mathrm{log}_{2}\left(1+ \gamma_{k,n} \right),
\label{eq:DLrate}
\end{align}
where $\gamma_{k,n}=\frac{S_{k,n} }{I_{k,n}}$ is the instantaneous downlink SINR at time $n$, and the expectation operates over all channel realizations.

\section{Performance Analysis}\label{Deterministic}
The investigation of the downlink SINR with  MRT  precoder accounting for the effects of imperfect and delayed CSIT due to user mobility and phase noise is the subject of this section. Specifically, we derive the deterministic equivalents of the SINR, being tight approximations  even for moderate system dimensions.  Simulations in Section~\ref{Simulations} corroborate our results and their tightness. 

The deterministic equivalent of the SINR $\gamma_{k,n}$ is such that $\gamma_{k,n}-\bar{\gamma}_{k,n}\xrightarrow[M \rightarrow \infty]{\mbox{a.s.}}0$\footnote{Note that $\xrightarrow[ M \rightarrow \infty]{\mbox{a.s.}}$ denotes almost sure convergence, and  $a_n\asymp b_n$ expresses the equivalence relation $a_n - b_n  \xrightarrow[ M \rightarrow \infty]{\mbox{a.s.}}  0$ with $a_n$  and $b_n$  being two infinite sequences.}, while the deterministic rate of user $k$ is obtained by the dominated convergence and the continuous mapping theorem~\cite{Vaart} as 
\begin{align}
R_{k}-\frac{1}{T_{c}}\sum_{n=1}^{T_{c}-\tau}\log_{2}(1 + \bar{\gamma}_{k,n}) \xrightarrow[ N \rightarrow \infty]{\mbox{a.s.}}0.\label{DeterministicSumrate}
\end{align}

The deterministic equivalent downlink achievable rate of $k$th user with MRT follows.
\begin{Theorem}\label{theorem:DLagedCSIMRT}
The deterministic equivalent of the downlink SINR with MRT precoding, accounting for imperfect CSI due to general impairments, and delayed CSI due to phase noise and user mobility is given by 
\begin{align}
\bar{\gamma}_{k,n} = \frac{  e^{-2\left( \sigma_{\varphi_{k}}^{2}+\sigma_{\phi}^{2} \right)n} {\delta}_{k}^{2}}{{\frac{1}{M}{\delta}_{k}^{'}}+  \frac{\sigma_{k}^{2}}{p_{\mathrm{d}}  \bar{\lambda}M}
 + \sum_{i\neq k}\frac{1}{M}\delta_{i}^{''}},\label{eq:DLdelayedCSIetaMRT1}
 \end{align}
 with 
  \begin{align}
  \bar{\lambda}&=\left( \frac1K\sum_{i=1}^{K}\frac{1}{M}\tr \bA_{n}^{2}\bD_{k} \right)^{-1},\nn
  \end{align}
${\delta}_{k}=\frac{1}{M}\tr\bA_{n}^{2}\bD_{k}$, ${\delta}_{k}^{'}=\frac{1}{M}\tr\bA_{n}^{2} \bD_{k} \left( \bR_{k} - \bA_{n}^{2}\bD_{k} \right)$, and $\delta_{i}^{''}=\frac{1}{M} \tr\bA_{n}^{2}\bD_{i} \bR_{k}$.
\end{Theorem}
\begin{proof} The proof of Theorem~\ref{theorem:DLagedCSIMRT} is given in Appendix~\ref{theorem2}.\end{proof}

%

%
\begin{proposition}\label{powerLaw}
 For the identical SLOs or CLO setting, the downlink achievable SINR of user $k$ with MRT precoding at time $n$ subject to imperfect CSI, delayed CSI due to phase noise  and user mobility, and collocated  BS antennas  becomes
 \begin{align}
  \gamma_{k,n}=\frac{\tau E_{\mathrm{d}}E_{\mathrm{u}}}{\sigma^{4}M^{2 q-1}}\mathrm{J}_{0}^{2}(2 \pi f_{\mathrm{D}}T_{\mathrm{s}}n)e^{-3\left( {\sigma_{\varphi_{k}}^{2}+\sigma_{\phi}^{2}} \right)n}[\bR_{k}^{2}]_{mm},\label{scalingPower}
 \end{align}
where the transmit uplink and downlink powers are scaled proportionally to $1/M^{q}$, i.e.,  $p_{\mathrm{u}} = E_{\mathrm{u}}/\sqrt{M}$ and $p_{\mathrm{d}}= E_{\mathrm{d}}/M^{q}$
for fixed $E_{\mathrm{u}}$ and $E_{\mathrm{d}}$, and $q>0$.
\end{proposition}
\begin{proof}
Let us first substitute the MRT precoder  in~\eqref{eq:DLreceivedSignalgeneral}. Then, we divide both the desired and the interference parts by $\lambda$ and $1/M^{2}$. The desired signal power is written as
\begin{align}
S_{k,n}^{\mathrm{MRT}}
&=  \frac{1}{M^{2}}\Big|\bg^\H_{k,n} \widetilde{\bTheta}_{k,n} \bA_{n} {\hat{\bg}}_{k,0}\Big|^{2}\nn\\
&=  \frac{1}{M^{2}}\Big|\hat{\bg}^\H_{k,0}\widetilde{\bTheta}_{k,n}  \bA_{n}^{2} {\hat{\bg}}_{k,0}\Big|^{2}\label{desiredMid}\\
&=\frac{1}{M^{2}}\mathrm{J}_{0}^4{2}(2 \pi f_{\mathrm{D}}T_{\mathrm{s}}n)e^{-4 \left( {\sigma_{\varphi_{k}}^{2}+\sigma_{\phi}^{2}} \right)n}[\bD_{k}]_{mm}^{2},\label{desired}
\end{align}
where $[\bD_{k}]_{mm}$ is the $m$th diagonal element of the matrix  $\bD_{k}$ expressing the variance of the $m$th element. Moreover, in the last step of~\eqref{desired}, we have used the law of large numbers for large $M$ in the case of collocated  BS antennas, as well as that $\bA_{n}$ is a scaled identity matrix  or a scalar in the case of identical SLOs or CLO, respectively. In other words, $ {\hat{\bg}}_{k,0}$ has i.i.d. elements with variance $[\bD_{k}]_{mm}$. As far as the interference is concerned, the first and third terms of~\eqref{eq:DLgenIntfPower} vanish to zero as $M \to \infty$ by means of the same law. Thus, we have
\begin{align}
I_{k,n}^{\mathrm{MRT}}&=\frac{\sigma^{2}_{k}}{M^{2}p_{\mathrm{d}} \lambda}\nn\\
&=\frac{\sigma^{2}_{k}\mathrm{J}_{0}^{2}(2 \pi f_{\mathrm{D}}T_{\mathrm{s}}n)e^{-\left( {\sigma_{\varphi_{k}}^{2}+\sigma_{\phi}^{2}} \right)n}[\bD_{k}]_{mm}}{M p_{\mathrm{d}} }, 
\end{align}
where  we have applied the law of large numbers to obtain $ \lambda\!=\!\left(\frac{1}{M}\mathrm{J}_{0}^{2}(2 \pi f_{\mathrm{D}}T_{\mathrm{s}}n)e^{-\left( {\sigma_{\varphi_{k}}^{2}\!+\sigma_{\phi}^{2}} \right)n}[\bD_{k}]_{mm}\! \right)^{-1}\!$ from~\eqref{eq:lamda}. Finally, we  substitute $\bD_{k}=\frac{\tau E_{\mathrm{u}}}{M^{q}\sigma^{2}_{k}}\bR_{k}^{2}$, and  the result in~\eqref{scalingPower} is obtained. \end{proof}

\begin{remark}
Clearly,~\eqref{scalingPower} shows that the achievable rate per user depends on the selection of  $q$. Specifically, if $q\le 1/2$, $\gamma_{k,n}$ is unbounded.  On the contrary, the SINR of user $k$ reduces to zero, if $q>1/2$. This clearly indicates that the transmit power of each user has been reduced over the required value. Most importantly, when $q=1/2$, the SINR approaches a finite non-zero limit given by 
\begin{align}
   \gamma_{k,n}=\frac{\tau^{2} E_{\mathrm{d}} E_{\mathrm{u}}}{\sigma_{k}^{4}}\mathrm{J}_{0}^{2}(2 \pi f_{\mathrm{D}}T_{\mathrm{s}}n)e^{-3\left( {\sigma_{\varphi_{k}}^{2}+\sigma_{\phi}^{2}} \right)n}[\bR_{k}^{2}]_{mm}.
\end{align}
\end{remark}

As a result, both phase noise and user mobility reduce the effective SINR, but neither of them changes the power scaling law.

\section{Numerical Results}\label{Simulations}
This section presents some representative numerical examples illustrating the impact of channel aging due to the combined effects of  phase noise and  user mobility on the performance of very large MIMO systems with MRT.   The correctness of the proposed results is validated by Monte-Carlo simulations. 

The simulation setup, based on long-term evolution (LTE) system specifications,  considers a single cell with radius of $R = 1000$ meters, and it assumes a guard range of $r_{0} = 100$ meters specifying the distance between the nearest user and the BS. The  BS, including $M$ antennas, broadcasts to $K$ users being uniformly distributed within the cell. While the theoretical analysis considers general $\bR_{k}$, this section takes into account only for large-scale fading (lognormal shadowing and path-loss), i.e., $\bR_{k}=\beta_{k}\Id_{M}$ with $\beta_{k}=z_{k}/\left( r_{k}/r_{0} \right)^{\upsilon}$, where $z_k$ is a log-normal random variable with standard deviation $\sigma$ ($\sigma = 8$ dB) expressing the shadow fading effect, $r_{k}$ denotes the distance between user $k$ and the BS, and $\upsilon$ ($\upsilon = 3.8$) is the path loss exponent. The  training duration consists of   $\tau=K$ symbols, and the
phase noises  at the  BS and user LOs are simulated as discrete Wiener processes by~\eqref{phaseNoiseBS}, with
increment standard deviations in the interval $0\deg-2\deg$~\cite{Colavolpe}. The coherence time is $T_{c}=1/4 f_{\mathrm{D}}=1~\mathrm{ms}$, where $f_{\mathrm{D}}=250~ \mathrm{Hz}$ is the Doppler spread corresponding to a relative velocity of  $135$ km/h between the BS and the users, if the  center frequency is assumed to be $f_c= 2~\mathrm{GHz}$. Moreover, given that the bandwidth for LTE-A is $\mathrm{W}=20\mathrm{MHz}$, the symbol time is $T_s=1/(2\mathrm{W})=0.025~\mathrm{\mu s}$. 

One useful metric  is the sum spectral efficiency denoted by
\begin{align} \label{eq num 1}
 \mathcal{S} \triangleq\sum_{k=1}^K
 \bar{R}_{k}.
\end{align}

In the following figures, the ``solid'', ``dash'', and ``dot'' lines designate the analytical results with no phase noise, $\phi_{m,n}=0\deg,
 \varphi_{k,n}=2\deg$, and $\phi_{m,n}= \varphi_{k,n}=2\deg$, respectively.   The bullets represent the simulation results.

Fig~\ref{fig:1} depicts the achievable sum-rate by varying the number of BS antennas for different values of phase noise in a  static environment ($v=0~\mathrm{Km/h}$). By increasing the hardware imperfection in terms of phase noise, the sum-rate decreases. When the phase noise at the user side has a variance of $\sigma_{\varphi_{k}}^{2}=2\deg$ the rate loss is $56\%$ with the BS having $M=30$ antennas, but when $M$ reaches $300$, the loss is smaller, i.e.,  $47\%$. An extra phase noise of $0\deg$ in the case of of CLO/ILOs setup leads to further losses of $50\%$  and  $47\%$, if $M=30$ and $M=300$, respectively.  

\begin{figure}[t]
    \centering
    \centerline{\includegraphics[width=0.49\textwidth]{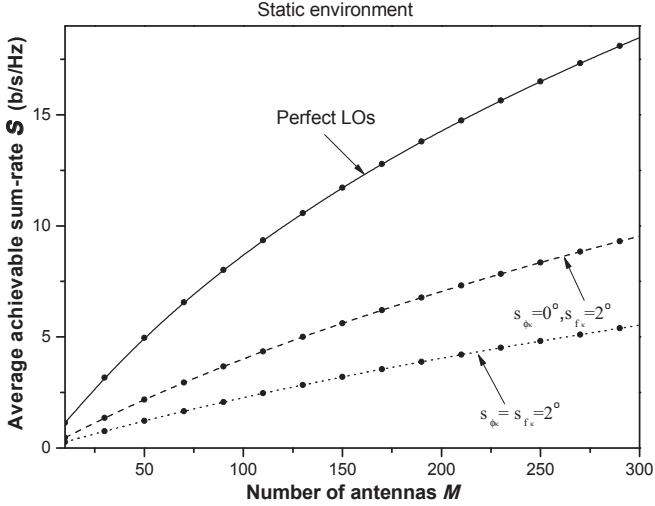}}
    \caption{Simulated and deterministic equivalent downlink sum-rates with MRT precoder in a static environment with respect to different number of BS antennas, and  various values of phase noise.}
    \label{fig:1}
\end{figure}

The variation of $\mathcal{S}$ with the normalized Doppler shift $f_{\mathrm{D}}T_{\mathrm{s}}$ is presented in Fig.~\ref{fig:2}, when $M=60$ for different values of phase noise starting from perfect LOs (no  phase noise) to high phase noise ($\sigma_{\phi_{k}}^{2}=\sigma_{\varphi_{k}}^{2}=2\deg$). The effect of Doppler shift dominates over the impact of phase noise. In fact, for low velocities in the order of $30~\mathrm{km/h}$ equivalent to $f_{\mathrm{D}}T_{\mathrm{s}}\approx 0.2$, the degradation due to phase noise starts to become  insignificant, but then the achievable rate can become so low that it is inadequate for practical applications.  In the same figure, the straight line, which is parallel to the horizontal axis, illustrates the sum-rate with imperfect CSI, but no channel aging (no user mobility or phase noise).  Furthermore,  as phase noise increases, the performance worsens. 

\begin{figure}[t]
    \centering
    \centerline{\includegraphics[width=0.49\textwidth]{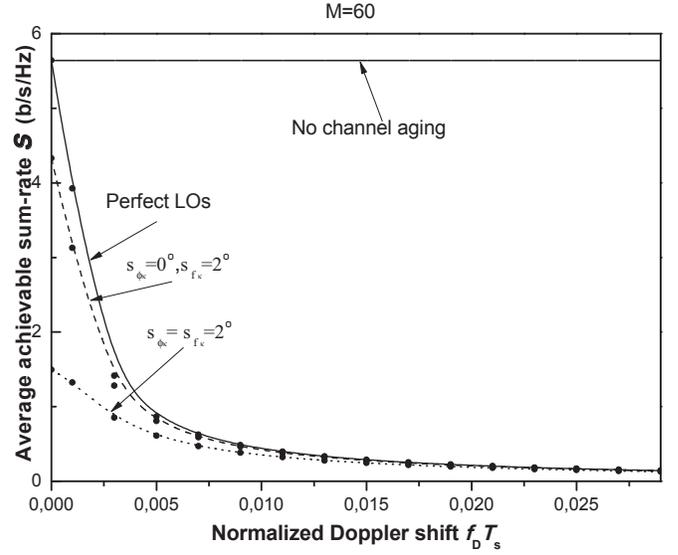}}
    \caption{Simulated and deterministic equivalent downlink sum-rates with MRT  precoder for $M=60$ with respect to the normalized Doppler shift,  and various values of phase noise.}
    \label{fig:2}
\end{figure}

Fig.~\ref{fig:3} illustrates how $p_{\mathrm{d}}$ varies with the number of BS antennas in a static environment ($f_{\mathrm{D}}T_{\mathrm{s}}=0$), in order to  achieve  $1$ bit/s/Hz. Specifically,  $p_{\mathrm{d}}$ decreases considerably when we increase $M$. Especially, a closer observation shows a reduction in the transmit power by approximately 1.5dB  after doubling the number of BS antennas, which agrees with previously known results e.g.,~\cite{Ngo1}. The more severe is the phase noise, the higher the downlink power is. 
\begin{figure}[t]
    \centering
    \centerline{\includegraphics[width=0.49\textwidth]{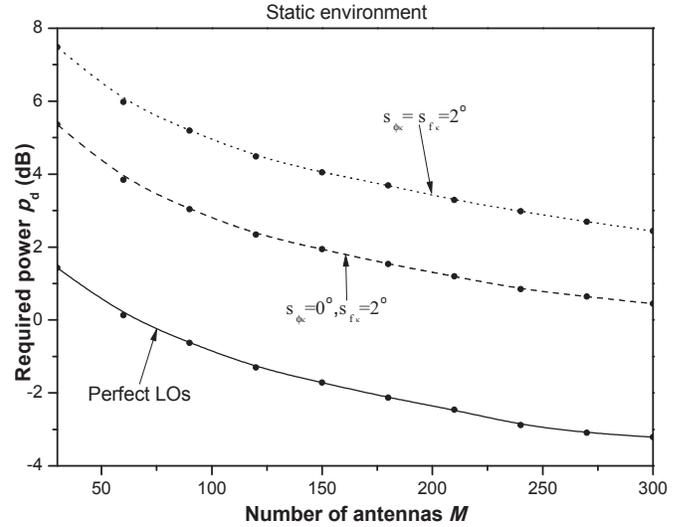}}
    \caption{Required transmit power  to achieve $1$ bit/s/Hz per user  with MRT  precoder in a static environment  with respect to different number of BS antennas,  and various values of phase noise.}
    \label{fig:3}
\end{figure}

\section{Conclusions}
Channel aging, being a limiting factor in massive MIMO systems, considered initially  only  user mobility. However, phase noise, which innovates with time, is more severe in large antenna arrays than in convetional systems, since the former include low quality antenna elements. As a result, its contribution to channel aging is substantial. Hence, we obtained an interesting joint channel-phase noise model, and we derived the deterministic equivalent of the downlink SINR. Notably, we showed that user mobility affects more system's performance than phase noise. Moreover, the phase noise plays a significant role on  the achievable rate in low mobility conditions. Interestingly, massive MIMO are still favorable under generalized channel aging.
\begin{appendices}

 \section{Proof of Theorem~\ref{theorem:DLagedCSIMRT}}\label{theorem2}
The derivation of the deterministic SINR starts by  deriving the deterministic equivalent of the normalization parameter $\lambda$. Specifically, the normalization parameter can be written by means of~\eqref{eq:lamda} as
\begin{align}
\lambda&=
\frac{K}{\EE\Big[\frac{1}{M}\tr \bA_{n}{\hat{\bG}}_{0}{\hat{\bG}}_{0}^{\H} \bA_{n}   \Big]} \nn\\
&= \frac{K}{\EE\Big[\sum_{i=1}^{K}\frac{1}{M}\bA_{n}^{2}  \hat{\bg}_{k,0} \hat{\bg}_{k,0}^{\H}  \Big]}\nn\\
&\asymp \left( \frac1K\sum_{i=1}^{K}\frac{1}{M}\tr \bA_{n}^{2}\bD_{k} \right)^{-1},\label{desired1MRT}
\end{align}
where we have applied~\cite[Lem. 1]{Truong}.
Similarly, regarding the rest part of the desired signal power, we obtain
\begin{align}
 \frac{1}{M}\EE\!\left[{\bg}^\H_{k,n}\widetilde{\bTheta}_{k,n}  \bA_{n} \hat{\bg}_{k,0} \right]
 &\!\!=\!\! \frac{1}{M}\EE\!\left[\!(\hat{\bg}^\H_{k,0}\bA_{n}  \!+ \!\tilde{\bee}^\H_{k,n} ) \widetilde{\bTheta}_{k,n}\bA_{n} \hat{\bg}_{k,0} \right]\nn\\
 &\!\!=\!\!\frac{1}{M}\EE\Big[{\hat{\bg}^\H_{k,0}\bA_{n}\widetilde{\bTheta}_{k,n} \bA_{n}\hat{\bg}_{k,0} }\Big],\nn\\
 &\!\!\asymp\!\!\frac{1}{M}\EE\left[\tr\bA_{n}\widetilde{\bTheta}_{k,n}\bA_{n}\bD_{k}\label{desired4MRT}\right]\\
 &\!\!\asymp \!\!\frac{e^{-\left( \sigma_{\varphi_{k}}^{2}+\sigma_{\phi}^{2} \right)n}}{M}\tr\bA_{n}^{2}\bD_{k}.\label{desired2MRT}
\end{align}

The deterministic equivalent signal power  $\bar{S}_{k,n} = \lim_{M \rightarrow \infty} S_{k,n}$ becomes by using~\eqref{desired2MRT} as
\begin{align}
\bar{S}_{k,n}\asymp  \bar{\lambda}e^{-2\left( \sigma_{\varphi_{k}}^{2}+\sigma_{\phi}^{2} \right)n}  {\delta}_{k}^{2},\label{eq:theorem4.5MRT}
\end{align}
where ${\delta}_{k}=\frac{1}{M}\tr\bA_{n}^{2}\bD_{k}$.

Now, we proceed with the derivation of each term  of~\eqref{eq:DLgenIntfPower}. Thus, the first term can be written as 
\begin{align}
&\frac{1}{M^{2}} \var\left[\bg^\H_{k,n} \widetilde{\bTheta}_{k,n} \bA_{n} \hatvg_{k,0} \right]\nn\\
&~~~~~~~~~~-\frac{1}{M^{2}} \EE\bigg[\Big|{\tilde{\bee}_{k,n}^{\H}\widetilde{\bTheta}_{k,n} \bA_{n}\hatvg
_{k,0} }\Big|^2\bigg]\xrightarrow[ M \rightarrow \infty]{\mbox{a.s.}} 0,
\end{align}
where the property of the variance operator $\mathrm{var}\left( x \right)=\EE[x^{2}]- \EE^{2}[x] $ together with~\eqref{eq:GaussMarkov2} have been used. Lemma~1 in~\cite{Truong} enables us to derive the deterministic equivalent as
\begin{align}
 \frac{\lambda}{M} \var\left[\bg^\H_{k,n} \widetilde{\bTheta}_{k,n}  \bA_{n} \hatvg_{k,0} \right] 
\asymp&  \frac{\bar{\lambda}}{M}{\delta}_{k}^{'}, \label{eq:theorem4.7MRT}
\end{align}
where ${\delta}_{k}^{'}=\frac{1}{M}\tr\bA_{n}^{2}\bD_{k} \left( \bR_{k} - \bA_{n}^{2}\bD_{k} \right)$.

Similarly, use of~\cite[Lem. 1]{Truong} in the last term completes the proof. Specifically, we have
\begin{align}
  \frac{\lambda}{M^{2}}\EE\Big[\left|\bg^\H_{k,n} \widetilde{\bTheta}_{k,n}\bA_{n} \hatvg_{i,0} \right|^{2}\Big] 
&\asymp \frac{\bar{\lambda}}{M^{2}} \tr\widetilde{\bTheta}_{k,n}\bA_{n}^{2}\bD_{i}\widetilde{\bTheta}_{k,n} ^{\H}\bR_{k},\nn\\
&=\frac{\bar{\lambda}}{M}\delta_{i}^{''},\label{lasttermMRT}
\end{align}
since $\bg_{k,n} $ and $\hatvg_{i,0} $ are mutually independent. Note here that $\delta_{i}^{''}=\frac{1}{M} \tr\bA_{n}^{2}\bD_{i} \bR_{k}$. 
\end{appendices}

\end{document}